\documentclass[runningheads]{llncs}
\usepackage[T1]{fontenc}
\usepackage{amsmath}
\usepackage{amssymb}
\usepackage{comment}
\usepackage{lmodern}
\usepackage{url}
\usepackage{graphicx}
\usepackage[linesnumbered]{algorithm2e}
\usepackage{todonotes}
\usepackage{mathtools}
\usepackage{arrayjob}
\usepackage{array}
\newcolumntype{M}[1]{>{\centering\arraybackslash}m{#1}}
\usepackage{xcolor,colortbl}
\usepackage{hyperref}

\usepackage[capitalize, nameinlink]{cleveref}

\newcommand{\wn}[2][inline]{\todo[color=green!50,#1]{{\textbf{W:} #2}}}

\newcommand{\jc}[2][inline]{\todo[color=brown!50,#1]{{\textbf{J:} #2}}}

\newcommand{\Oh}{\mathcal{O}}
\renewcommand{\leq}{\leqslant}
\renewcommand{\geq}{\geqslant}

\renewcommand{\ge}{\geqslant}

\title{Dynamic data structures for twin-ordered matrices}

 \author{
     Bartłomiej Bosek\inst{1,2}\thanks{\includegraphics[height=1.2em]{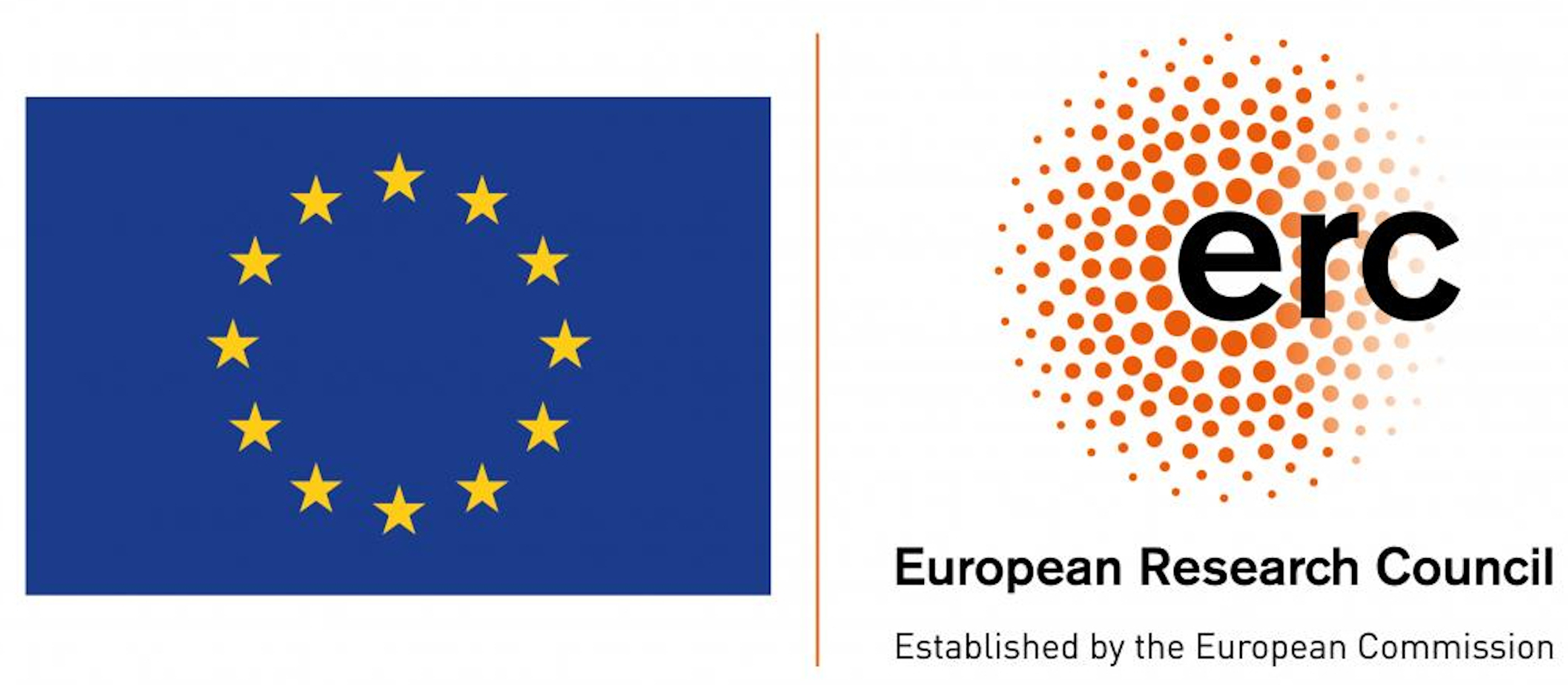}
     {\em{B. Bosek, W. Nadara, Mi. Pilipczuk, A. Zych-Pawlewicz:}}
     This work is a~part of project BOBR that has received funding from the European Research Council (ERC) under the European Union’s Horizon 2020 research and innovation programme (grant agreement No. 948057).} \and
     Jadwiga Czyżewska\inst{2}\thanks{Supported by Polish National Science Centre SONATA BIS-12 grant number 2022/46/E/ST6/00143.} \and
     Evangelos Kipouridis\inst{3} \and
     Wojciech Nadara\inst{2}\inst{\star} \and
     Michał Pilipczuk\inst{2}\inst{\star} \and
     Karol Węgrzycki\inst{3}\thanks{Supported by 
    the Deutsche Forschungsgemeinschaft (DFG, German Research Foundation) grant number
    559177164.} \and
     Anna Zych-Pawlewicz\inst{2}\inst{\star}
 }

 \authorrunning{B. Bosek, J. Czyżewska, E. Kipouridis, W. Nadara, M. Pilipczuk, K. Węgrzycki and A. Zych-Pawlewicz}

 \institute{
     Institute of Theoretical Computer Science, Jagiellonian University, Kraków, Poland\\
     \email{bartlomiej.bosek@uj.edu.pl} \and
     Institute of Informatics, University of Warsaw, Poland\\
     \email{\{j.czyzewska, w.nadara, michal.pilipczuk, anka\}@mimuw.edu.pl} \and
     Max Planck Institute for Informatics, Saarland Informatics Campus, Germany\\
     \email{\{kipouridis, kwegrzyc\}@mpi-inf.mpg.de}
 }

\begin{document}
\maketitle

\begin{abstract}
We present a~dynamic data structure for representing binary $n\times n$ matrices
that are $d$-twin-ordered, for a~fixed parameter $d$. Our structure
supports cell queries and single-cell updates both in $\Oh(\log \log n)$ expected worst case
time, while using
$\Oh_d(n)$ memory; here,
the $\Oh_d(\cdot)$ notation hides a~constant depending on $d$.
\end{abstract}

\clearpage
\setcounter{page}{1}
\newpage

\section{Introduction}\label{sec:introduction}

The notion of \emph{twin-width}, introduced by Bonnet, Kim, Thomass\'e, and Watrigant~\cite{tww1}, has rapidly emerged as a unifying structural parameter for graphs (both sparse and dense), matrices, and permutations. The basic idea behind twin-width, expressed through the underlying decomposition notion called a {\em{contraction sequence}}, is to ``fold'' the given structure (be it a graph, a matrix, or a~permutation) by iteratively merging elements into larger and larger parts. At~each point, we require that every part has a homogeneous relation to all the other parts, except for a bounded number of exceptions. The best possible bound on the number of exceptions is then precisely the twin-width. This concept has a~surprisingly wide range of applicability: for example, $H$-minor-free graphs (for a fixed $H$), graphs of bounded cliquewidth, and pattern-free permutations all have bounded twin-width~\cite{tww1}.

Despite being a relatively new concept, the fundamental character of twin-width has already been well-understood. Twin-width uncovers deep links between structural graph theory and algorithmic finite model theory, see e.g.~\allowbreak\cite{tww-separable,tww8,tww4,tww1,tww-perm,tww-types,tww-stable,tww-tournaments,tww-distal}. In this direction, probably the most important result is the theorem of Bonnet et al.~\cite{tww1} stating that the model-checking problem for first-order logic is fixed-parameter tractable on graph classes of bounded twin-width, provided a suitable contraction sequence is given as input. Combinatorial aspects of twin-width were explored in~\cite{tww2,tww7,tww4,tww6,tww-pchi,tww-qchi}, among many other works. Several further algorithmic applications can be found in~\cite{tww5,tww-kern,tww-repr}.



In this work, we study binary matrices of bounded twin-width -- let us now define the notion precisely. We consider a binary matrix $M$, i.e., one with entries in $\{0,1\}$. A contraction sequence for $M$ is a sequence of pairs 
$$(\mathcal{R}_0, \mathcal{C}_0),\ (\mathcal{R}_1, \mathcal{C}_1), \dots ,\ (\mathcal{R}_p, \mathcal{C}_p),$$
where for every $i=0,1,\dots, p$ we have $\mathcal{R}_i$ and $\mathcal{C}_i$ are partitions of rows and columns of $M$ into row blocks and column blocks respectively. Moreover,
\begin{itemize}
	\item in $(\mathcal{R}_0, \mathcal{C}_0)$ each row and each column belongs to a separate block;
	\item in $(\mathcal{R}_p, \mathcal{C}_p)$ all rows belong to a~single block and all columns belong to a single block;
	\item $(\mathcal{R}_i, \mathcal{C}_i)$ is obtained from  $(\mathcal{R}_{i-1}, \mathcal{C}_{i-1})$ by either merging two consecutive row blocks into a single row block or merging two consecutive column blocks into a single column block.
\end{itemize}
For an example of a contraction sequence, see Figure~\ref{fig:contraction_sequence}.
 For a row block $R\in \mathcal{R}$ and a column block $C\in \mathcal{C}$, the {\em{zone}} $M[R,C]$ is the submatrix consisting of all the cells at the intersection of the rows of $R$ and the columns of $C$. We say that a zone is {\em{non-constant}} if it contains both an entry $0$ and an entry $1$. The {\em{width}} of the contraction sequence is the least $d$ such that at every point in the sequence, in every row block and in every column block there are at most $d$ non-constant zones. We say that $M$ is {\em{$d$-twin-ordered}} if it admits a contraction sequence of width at most $d$, and $M$ has {\em{twin-width}} at most $d$ if one can permute rows and columns in order to obtain a $d$-twin-ordered matrix. Throughout this paper we consider the order of the rows and columns to be fixed, and we always assume that the considered matrix is~$d$-twin-ordered. 

\begin{figure}[h]
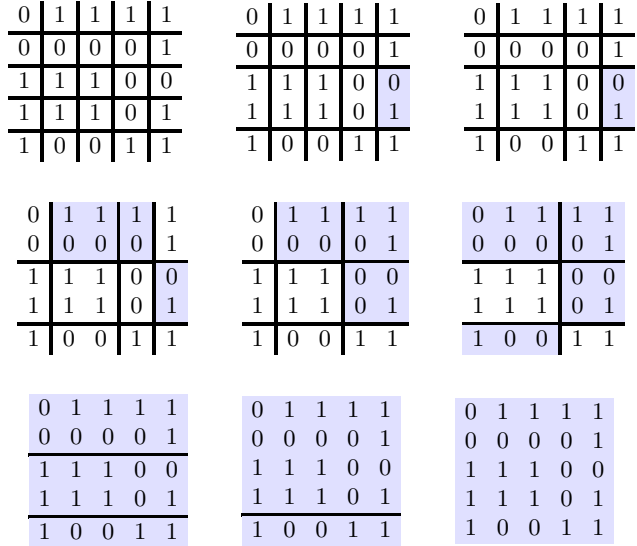

	\newcommand{\g}{\cellcolor{blue!12!white}}
	\setlength{\arrayrulewidth}{.15em}
	\centering

	\begin{tabular}{M{1em}|M{1em}|M{1em}|M{1em}|M{1em}}
		
			0 & 1 & 1 & 1 & 1 \\
		\hline
			0 & 0 & 0 & 0 & 1 \\
		\hline
			1 & 1 & 1 & 0 & 0 \\
		\hline
			1 & 1 & 1 & 0 & 1 \\
		\hline
			1 & 0 & 0 & 1 & 1 \\		
	\end{tabular}
	\hspace{0.5cm}
	\begin{tabular}{M{1em}|M{1em}|M{1em}|M{1em}|M{1em}}
		0 & 1 & 1 & 1 & 1 \\
		\hline
		0 & 0 & 0 & 0 & 1 \\
		\hline
		1 & 1 & 1 & 0 & \g 0 \\
		1 & 1 & 1 & 0 & \g 1 \\
		\hline
		1 & 0 & 0 & 1 & 1 \\		
	\end{tabular}
	\hspace{0.5cm}
	\begin{tabular}{M{1em}|M{1em} M{1em}|M{1em}|M{1em}}
		0 & 1 & 1 & 1 & 1 \\
		\hline
		0 & 0 & 0 & 0 & 1 \\
		\hline
		1 & 1 & 1 & 0 & \g 0 \\
		1 & 1 & 1 & 0 & \g 1 \\
		\hline
		1 & 0 & 0 & 1 & 1 \\		
	\end{tabular}
	\\
	\vspace{0.5cm}
	\begin{tabular}{M{1em}|M{1em} M{1em}|M{1em}|M{1em}}
		0 & \g 1 & \g 1 & \g 1 & 1 \\
		0 & \g 0 & \g 0 & \g 0 & 1 \\
		\hline
		1 & 1 & 1 & 0 & \g 0 \\
		1 & 1 & 1 & 0 & \g 1 \\
		\hline
		1 & 0 & 0 & 1 & 1 \\		
	\end{tabular}
	\hspace{0.5cm}
	\begin{tabular}{M{1em}|M{1em} M{1em}|M{1em} M{1em}}
		0 & \g 1 & \g 1 & \g 1 & \g 1 \\
		0 & \g 0 & \g 0 & \g 0 & \g 1 \\
		\hline
		1 & 1 & 1 & \g 0 & \g 0 \\
		1 & 1 & 1 & \g 0 & \g 1 \\
		\hline
		1 & 0 & 0 & 1 & 1 \\		
	\end{tabular}
	\hspace{0.5cm}
	\begin{tabular}{M{1em} M{1em} M{1em}|M{1em} M{1em}}
		\g 0 & \g 1 & \g 1 & \g 1 & \g 1 \\
		\g 0 & \g 0 & \g 0 & \g 0 & \g 1 \\
		\hline
		1 & 1 & 1 & \g 0 & \g 0 \\
		1 & 1 & 1 & \g 0 & \g 1 \\
		\hline
		\g 1 & \g 0 & \g 0 & 1 & 1 \\		
	\end{tabular}
	\\
	\vspace{0.5cm}
	\begin{tabular}{M{1em} M{1em} M{1em} M{1em} M{1em}}
		\g 0 & \g 1 & \g 1 & \g 1 & \g 1 \\
		\g 0 & \g 0 & \g 0 & \g 0 & \g 1 \\
		\hline
		\g 1 & \g  1 & \g 1 & \g 0 & \g 0 \\
		\g 1 & \g  1 & \g 1 & \g 0 & \g 1 \\
		\hline
		\g 1 & \g 0 & \g 0 & \g  1 & \g 1 \\		
	\end{tabular}
	\hspace{0.5cm}
	\begin{tabular}{M{1em} M{1em} M{1em} M{1em} M{1em}}
		\g 0 & \g 1 & \g 1 & \g 1 & \g 1 \\
		\g 0 & \g 0 & \g 0 & \g 0 & \g 1 \\
		\g 1 & \g  1 & \g 1 & \g 0 & \g 0 \\
		\g 1 & \g  1 & \g 1 & \g 0 & \g 1 \\
		\hline
		\g 1 & \g 0 & \g 0 & \g  1 & \g 1 \\	
	\end{tabular}
	\hspace{0.5cm}
	\begin{tabular}{M{1em} M{1em} M{1em} M{1em} M{1em}}
		\g 0 & \g 1 & \g 1 & \g 1 & \g 1 \\
		\g 0 & \g 0 & \g 0 & \g 0 & \g 1 \\
		\g 1 & \g  1 & \g 1 & \g 0 & \g 0 \\
		\g 1 & \g  1 & \g 1 & \g 0 & \g 1 \\
		\g 1 & \g 0 & \g 0 & \g  1 & \g 1 \\	
	\end{tabular}
	\let\g\undefined
	\caption{A contraction sequence of width $2$: if we merge two consecutive rows or columns, we erase a line separating them. For example, in the first step we merge the third and the fourth row, in the second we merge the second and the third column. 
		Non-constant zones are colored: note that at any point of the sequence any row block and column block contains at most two non-constant zones.}
	\label{fig:contraction_sequence}
\end{figure}

The connection between graphs and matrices of bounded twin-width lies at the heart of the theory of twin-width and was observed already in the original work of Bonnet et al.~\cite{tww1}. Later, Bonnet, Giocanti, Ossona de Mendez, Simon, Thomass\'e, and Toru\'nczyk studied the twin-width of matrices from the combinatorial and the model-theoretic viewpoints~\cite{tww4}, whereas Bonnet, Giocanti, Ossona de Mendez, and Thomass\'e investigated algorithmic applications of twin-width for problems such as matrix multiplication~\cite{tww5}. The matrix viewpoint was also pivotal in the work on $\chi$-boundedness of graphs of bounded twin-width~\cite{tww-pchi,tww-qchi}.

The starting point of our work is the result of Pilipczuk, Sokołowski, and Zych-Pawlewicz~\cite{tww-repr}, who approached the setting of matrices of bounded twin-width from the point of view of data structures, or {\em{compact representations}}. Precisely, they designed a data structure that stores a $d$-twin-ordered binary $n\times n$ matrix $M$ using $\Oh_d(n)$\footnote{The $\Oh_d(\cdot)$ notation hides multiplicative factors that may depend on $d$.} bits (which is asymptotically optimal) while supporting single-entry queries in time $\Oh_d(\log \log n)$. Their data structure, however, is fundamentally static: once it is initialized on a matrix $M$, it cannot be easily modified upon updates to $M$. It is natural to ask whether a similarly compact and efficient data structure can be also proposed in the dynamic setting, where $M$ can be updated subject to a guarantee that it always stays $d$-twin-ordered.


\paragraph*{Our results.} In this work, we address the above question and present a dynamic data structure for $d$-twin-ordered binary matrices that supports both queries and updates efficiently, while guaranteeing near-linear space complexity. Formally, we show the following theorem.

\begin{theorem}\label{thm:main}
	Let $d\in \mathbb{N}$ and $M \in \{0,1\}^{n \times n}$ be a binary matrix
	that is updated over time subject to a guarantee that $M$ is
	$d$-twin-ordered at every moment. Then, there exists a dynamic data
	structure for maintaining $M$ that uses $\Oh_d(n)$ memory and supports
	the following operations:
	\begin{itemize}
		\item $\mathtt{Init}(n, \cal K)$: initialize the data structure
		with a list $\cal K$ of disjoint rectangles representing all ones
		in $M$ in $\Oh_d((n + |{\mathcal{K}}|) \log \log n)$ time.
		\item $\mathtt{QueryBit}(i,j)$: return the
		entry $M[i, j]$ in $\Oh(\log \log n)$ expected worst case
		time.
		\item $\mathtt{Update}(i, j)$: flip the entry $M[i, j]$
		in $\Oh(\log \log n)$ expected worst case time.
	\end{itemize}
\end{theorem}

Throughout the paper we work in the standard word RAM model with words
of length $\Oh(\log n)$, hence space complexity $\Oh_d(n)$ means storing
the matrix using \linebreak$\Oh_d(n \log n)$ bits. 
In the static setting,
Pilipczuk et al.~\cite{tww-repr} managed to obtain a~representation that
uses only $\Oh_d(n)$ bits, hence our result trades a slight increase in the space complexity for making the data structure dynamic.
Note that only the space complexity of the data structure and
its initialization time depend on the parameter $d$, whereas the
bounds on the time complexity of all the other operations do not
involve $d$.
Following~\cite{tww-repr} we allow our data structure to be
initialized on a compact description of the matrix through a list
of disjoint rectangles comprising all the entries $1$ (which we formally define in Section~\ref{sec:preliminaries} and call {\em{slabs}}); so in
particular, one may efficiently initialize a dense matrix. 
It is known (see \cite[Lemma~10]{tww-repr} for details) that such a
description with $\Oh_d(n)$ rectangles can be easily extracted
from a contraction sequence of width $d$.

\section{Preliminaries}
\label{sec:preliminaries}

All logarithms in this paper are base $2$. We assume the standard RAM model of computation. For $a,b\in \mathbb{N}$ with $a\leq b$, we define the {\em{segment}} $[a,b]\coloneqq \{a,a+1,\ldots,b\}$.

\paragraph*{Twin-ordered matrices.}
The basic terminology for matrices as well as the definition of
a $d$-twin-ordered matrix have already been discussed in
Section~\ref{sec:introduction}.
For the rest of the paper we work with a~binary $n\times n$ matrix $M$
that is at every point $d$-twin-ordered. The entry of $M$ at row $i$ and column $j$ is denoted by $M[i,j]$.




In our data structure we do not rely on contraction sequences at all,
but rather exploit structural properties of matrix $M$ following from
the assumption that it is $d$-twin-ordered. To state these properties,
we need a~few definitions. Recall that a~{\em{(row/column) block}} is
a contiguous set of rows/columns, and for a~row block $R$ and a~column
block $C$, the {\em{zone}} $M[R,C]$ is the submatrix consisting of all
the cells at the intersection of $R$ and $C$.

A {\em{corner}} in $M$ is a~$2\times 2$ zone whose both rows are different and both columns are different. We have the following fact about the number of corners in a $d$-twin-ordered matrix.

\begin{lemma}[{\cite[Theorem 5 and Lemma~8]{tww-repr}}]\label{lem:mixzones}
	Let $f_d=\frac{16}
{3}(2d + 3)^2 2^{4(2d+2)}$.  Any $d$-twin-ordered binary $n \times n$ matrix
	contains at most $f_d \cdot (n+2) \in \Oh_d(n)$ corners.
\end{lemma}

%

For the purpose of this paper, a~\emph{slab} in
$M$ is defined as a~zone $M[R,C]$ entirely filled with ones, where $R$ is a~row block and $C$ is a~column block. We can also represent a~slab as
$M[[a,b],[c,d]]$ by
the four coordinates $(a,b,c,d)$, with
$a,b,c,d \in \{ 1, \ldots n \}$ and $a\leq b$ and $c\leq d$. This is in particular the representation that we use when working with slabs in algorithmic procedures.
A {\em{slab decomposition}} of $M$
is a~set $\mathcal{K}$ of pairwise disjoint slabs in $M$
such that there is no
entry $1$ outside of the slabs of $\mathcal{K}$. See Figure~\ref{fig:slab_decomposition_example_prelims} for an example.

\begin{figure}[h]
    \newcommand{\h}{\cellcolor{red!22!white}}
	\newcommand{\z}{\cellcolor{gray!22!white}}
	\newcommand{\g}{\cellcolor{white!42!white}}
	\newcommand{\y}{\cellcolor{yellow!62!white}}
	\newcommand{\q}{\cellcolor{green!32!white}}
	\newcommand{\e}{\cellcolor{black!32!white}}
	\renewcommand{\arraystretch}{1.2}
	\centering
	\begin{tabular}{|M{1em}|M{1em}|M{1em}|M{1em}|M{1em}|}
		\hline
		\g 0 & \y 1 & \y 1 & \g 0 & \z 1 \\
		\hline
		\g 0 & \g 0 & \g 0 & \g 0 & \z 1 \\
		\hline
		\g 0 & \q  1 & \q 1 & \q 1 & \g 0 \\
		\hline
		\h 1 & \h  1 & \h 1 & \g 0 & \e 1 \\
		\hline
		\h 1 & \h  1 & \h 1 & \g 0 & \e 1\\	
		\hline
	\end{tabular}
	\let\h\undefined
	\let\g\undefined
	\let\z\undefined
	\let\y\undefined
	\let\e\undefined
	\let\q\undefined
	\caption{A slab decomposition of a matrix in the form of a set $\mathcal{K} = \{M[[1,1], [2,3]],\ M[[1,2], [5,5]],\ M[[3,3], [2,4]],\ M[[4,5], [1,3]],\ M[[4,5], [5,5]]\}.$ All slabs are marked with distinct colors.}
	\label{fig:slab_decomposition_example_prelims}
\end{figure}

\begin{lemma}[{\cite[Lemma~10]{tww-repr}}]\label{lem:rectangles}
Let $M$ be an $n \times n$ binary matrix that is $d$-twin-ordered. Then
$M$ admits a slab decomposition $\mathcal{K}$ with
$|\mathcal{K}| \leq d(2n-2)+1$.
\end{lemma}

We conclude this part with a summary of a~static data stracture, which was described by Chan in~\cite{Chan2DPL}.
We rephrase his result using the terminology of our paper:
\begin{theorem}[Chan~\cite{Chan2DPL}]\label{thm:Chan} Given a~set of $N$ disjoint slabs with coordinates belonging to $\{1,2,\dots, U\}$, we can build $\Oh(N)$-space data structure, such that given a~pair $(i,j)$, it returns the slab containing $(i,j)$ or returns an information that there is no such slab in  $\Oh(\log \log U)$ time. The preprocessing time is $\Oh(N \log \log U)$.
\end{theorem}
Let us call the query described in the statement of Theorem~\ref{thm:Chan} a {\em point location query} and the data structure {\em a~static 2D orthogonal point location data}.

\paragraph*{Dictionaries.}
In this work we make extensive use of \emph{van Emde Boas} dictionaries,
called also {van Emde Boas trees}~\cite{VANEMDEBOAS197780}. The
van Emde Boas tree
maintains a~set of keys $S$, which is a~subset of the
domain $\{ 1, \ldots, U \}$ for $U\in \mathbb N$, by supporting the following
operations:
\begin{itemize}
 \item $\mathtt{Insert}(\mathtt{x})$:
    insert key $\mathtt{x}$ to $S$;
 \item $\mathtt{Delete}(\mathtt{x})$: remove key $\mathtt{x}$
 from $S$;
 \item $\mathtt{Lookup}(\mathtt{x})$: verify whether $\mathtt{x}$
 belongs to the maintained set of keys $S$;
 \item $\mathtt{Successor}(\mathtt{x})$: find the smallest key in
 $S$ greater than $\mathtt{x}$;
 \item $\mathtt{Predecessor}(\mathtt{x})$: find the largest key
 in $S$ smaller than $\mathtt{x}$.
\end{itemize}
All operations work in worst case time $\Oh(\log \log U)$ and
the
required space is $\Oh(U)$ \cite{VANEMDEBOAS197780}.

Recall that a~{\em{segment}} $[a,b]$, for $a,b \in \mathbb{N}, a~\leq b$, is defined as the
set of consecutive natural numbers between its coordinates $a$ and $b$:
$[a,b]= \{a, a+1, \ldots, b-1, b \}$. Two segments $[a,b]$ and
$[a',b']$ are called
{\em{adjacent}} if either $b+1=a'$ or $a-1=b'$.
In the next lemma we describe a modification of van Emde Boas trees, which gives
us the data structure to maintain a~set of
disjoint and non-adjacent segments with coordinates in $\{ 1, \ldots, n \}$, which we shall call an {\em{adhesive segment set}}.


\begin{lemma}\label{lem:segment_tree}
Let $n \in \mathbb{N}$. There is an \emph{adhesive segment set} data
structure $\mathbb{S}$ that
maintains a
dynamic set $S$ of
pairwise disjoint and pairwise non-adjacent segments
whose coordinates belong to $\{ 1, \ldots, n \}$.
The data structure supports
the following methods:
\begin{itemize}
	\item $\mathtt{Containing}([a,b])$: Returns the segment
	of $S$ that contains
	$[a,b]$, if existent. 
	\item $\mathtt{Adjacent}([a,b])$: Returns the (at most two)
	segments in $S$
	adjacent to $[a,b]$.
	\item $\mathtt{Disjoint}([a,b])$: Returns whether $[a,b]$ is disjoint from all the segments in $S$.
 \item $\mathtt{Merge}([a,b])$: If $[a,b]$ is not disjoint from all the segments in $S$, the method does nothing, and otherwise it performs the following modification. If $[a,b]$ is not adjacent to any segment in $S$, the method adds $[a,b]$ to $S$. If $[a,b]$ is adjacent to one segment $s\in S$, the method replaces $s$ with $s'\coloneqq s\cup [a,b]$. If $[a,b]$ is adjacent to two segments $s_1,s_2\in S$, the method replaces $s_1$ and $s_2$ with $s'\coloneqq [a,b] \cup s_1 \cup s_2$.
 \item $\mathtt{Split}([a,b])$: If $[a,b]$
 is contained
 in some segment $[a',b'] \in S$, the method
 replaces the segment $[a',b']$ in
 $S$ with segments $[a',a-1]$ and $[b+1,b']$. If any of these segments turns out to be empty, it is not added to $S$. If there is no segment $[a',b']\in S$ containing $[a,b]$, the method does nothing.
\end{itemize}
The data structure occupies $\Oh(n)$ memory and all operations
run in $\Oh(\log \log n)$ worst case time. The data structure can be initialized on an empty $S$ in time $\Oh(n)$.
\end{lemma}
\begin{proof}
We use the van Emde Boas dictionary $\mathbb{D}$ with universe $[1,n]$. We store in $\mathbb{D}$
the coordinates of the segments of $S$. Additionally, for each coordinate we store an extra bit determining if it is the first or the
second coordinate. These extra bits are stored in a static array with entries $[1,n]$: entry $i$ stores the suitable piece of data for the coordinate equal to $i$.

To implement $\mathtt{Containing}([a,b])$ we issue
$a'\coloneqq \mathtt{Predecessor}(a+1)$ and
$b'\coloneqq \mathtt{Successor}(b-1)$ to $\mathbb{D}$. We
return $[a',b']$, provided $a'$ is the first coordinate of a~segment, $b'$ is the second coordinate of a~segment, and $b'=\mathtt{Successor}(a')$; otherwise there is no segment in $S$ containing $[a,b]$. Queries $\mathtt{Adjacent}([a,b])$ and $\mathtt{Disjoint}([a,b])$ are implemented
analogously.

To implement $\mathtt{Merge}([a,b])$, we first check whether $[a,b]$ is disjoint from all the segments in $S$ by invoking $\mathtt{Disjoint}([a,b])$. 
If that is not the case, we do nothing.
Otherwise, we issue
$a'\coloneqq\mathtt{Predecessor}(a)$ and
$b'\coloneqq\mathtt{Successor}(b)$ queries to $\mathbb{D}$.
Since $[a,b]$ is disjoint from segments in $S$, $a'$ is the second
and $b'$ is the first coordinate of a~segment in $S$. If $a'+1=a$,
we issue $a'':=\mathtt{Predecessor}(a')$ query to
$\mathbb{D}$ to find the first coordinate of the segment
$[a'',a'] \in S$. Then, we delete $a'$ from $\mathbb{D}$ and
insert $b$, indicating that this is the second coordinate of a~segment. 
Then we
proceed with $b'$ in an analogous fashion.

To implement $\mathtt{Split}([a,b])$, we issue
$[a',b']=\mathtt{Containing}([a,b])$ to find
the interval
$[a',b']\in S$ containing $[a,b]$. If there is no such interval, there is nothing to do.
Otherwise, we insert to
$\mathbb{D}$ key $a-1$ indicating that it is the second coordinate,
and we insert key $b+1$ indicating that it is the first coordinate.
\qed


\end{proof}

\section{Canonical slab decomposition}
\label{sec:decomposition}


From Lemma~\ref{lem:rectangles} we know that for any $d$-twin-ordered $n \times n$ matrix there always exists a slab decomposition into $\Oh_d(n)$ slabs, which is a crucial property of our dynamic data structure. However, we need to provide an efficient algorithm to compute such decomposition. 
In this Section, we describe the canonical slab decomposition, which is a slab decomposition with additional structural properties. This allows us to use the adhesive segment set data structure, described in Lemma~\ref{lem:segment_tree}, to compute a desired decomposition in Theorem~\ref{thm:canonical}.

\newcommand{\strips}{\mathtt{strips}}

The slab decomposition we are going to compute is, in some sense, canonical. To describe it, we need a~few definitions. Recall that $M$ is a~binary $n\times n$ matrix. For $i\in [1,n]$, an {\em{$i$-strip}} is a~slab of the form $(a,b,i,i)$ for some $1\leq a\leq b\leq n$ such that neither $(a-1,b,i,i)$ nor $(a,b+1,i,i)$ is a~slab. In other words, an $i$-strip is a~slab of width $1$ contained in the $i$-th column such that it is not possible to extend this slab upwards or downwards, due to either finding an entry $0$ there or the border of the matrix $M$. 
By $\strips_i$ we denote the set of all $i$-strips; and $\strips\coloneqq \bigcup_{1\leq i\leq n} \strips_i$ is the set of all the strips. For example, if $(0, 1, 1, 0, 1, 0, 1, 1, 1)$ are the consecutive entries of the $ith$ column of a matrix, the set $\strips_i$ is equal to $\{(2,3,i,i), (5,5, i, i), (7,9, i,i)\}$.

We say that strips $(a,b,i,i)\in \strips_i$ and $(a,b,j,j)\in \strips_j$ are {\em{siblings}} if for every $k\in [i,j]$ we have $(a,b,k,k)\in \strips_k$.
Note that the relation of being a~sibling is an equivalence relation on $\strips$. We call the equivalence classes of this relations {\em{canonical slabs}}; we identify each canonical slab with its union, which is always a~slab in $M$. Noting that each strip belongs to exactly one canonical slab, we define the {\em{canonical slab decomposition}} of $M$ as the family of all canonical slabs of $M$. For examples of the canonical slab decomposition see Figure~\ref{fig:slab_decomposition_example}.

\begin{figure}[h]
    \newcommand{\h}{\cellcolor{red!22!white}}
	\newcommand{\z}{\cellcolor{gray!22!white}}
	\newcommand{\g}{\cellcolor{blue!42!white}}
	\newcommand{\y}{\cellcolor{yellow!62!white}}
	\newcommand{\q}{\cellcolor{green!32!white}}
	\newcommand{\e}{\cellcolor{black!32!white}}
	\renewcommand{\arraystretch}{1.2}
	\centering
	\begin{tabular}{|M{1em}|M{1em}|M{1em}|M{1em}|M{1em}|}
		\hline
		\q 1 & \z 1 & \g 1 & \h 1 & \h 1 \\
		\hline
		   0 & \z 1 &    0 & \h 1 & \h 1 \\
		\hline
		\e 1 & \z 1 &    0 &    0 & 0 \\
		\hline
		\e 1 & \z 1 & \y 1 & \y 1 & \y 1 \\
		\hline
		\e 1 &    0 &    0 &    0 &  0 \\	
		\hline	
	\end{tabular}
	\hspace{0.5cm}
	\begin{tabular}{|M{1em}|M{1em}|M{1em}|M{1em}|M{1em}|}
		\hline
		   0 & \y 1 & \y 1 & 0 & 0 \\
		\hline
		   0 & \y 1 & \y 1 & 0 & 0 \\
		\hline
		\e 1 & \y 1 & \y  1 & \h  1 & 0 \\
		\hline
		\e 1 & \y 1 & \y 1 & \h 1 & 0 \\
		\hline
		 0 &    0 &    0 &    0 &  0 \\		
		\hline
	\end{tabular}
	
	\let\h\undefined
	\let\g\undefined
	\let\z\undefined
	\let\y\undefined
	\let\e\undefined
	\let\q\undefined
	\caption{The canonical slab decomposition of a matrix -- all canonical slabs are marked with distinct colors. Note that in the second matrix the slab decomposition may contain just two slabs $(3,4,1,4)$, $(1,2, 2,3)$, while in the canonical slab decomposition three slabs are required.}
	\label{fig:slab_decomposition_example}
\end{figure}

In the next lemma we show that the canonical slab decomposition of a~$d$-twin-ordered matrix is always~small.

\begin{lemma}\label{lem:canonical-small}
	Let $\cal R$ be the canonical slab decomposition of a~binary $n\times n$ matrix $M$ that is $d$-twin-ordered. Then $|{\cal R}|\leq 4 f_d \cdot (n+2)+4n \in \Oh_d(n)$, where $f_d$ is 
	the~constant dependent on $d$ defined in Lemma~\ref{lem:mixzones}.
\end{lemma}
\begin{proof}
	Let $(a,b,c,d)$ be a~canonical slab in $M$ with $1<a\leq b< n$ and $1<c\leq d<n$.
	Our goal is to show that $(a,b,c,d)$ must intersect a~corner in $M$. This will prove the claim, as there can be at most $4n$ canonical slabs satisfying $a=1$ or $b=n$ or $c=1$ or $d=n$, whereas the number of corners in $M$ is at most $f_d(n+2)$ by Lemma~\ref{lem:mixzones} and each corner in $M$ can intersect at most $4$ slabs of $\cal R$.
	
	Let $Z=M[[a,b],[d+1,d+1]]$.
	As $(a,b,c,d)$ is a~canonical slab, we have $(a,b,d+1,d+1)\notin \strips_{d+1}$; that is, $Z$ is not a~$(d+1)$-strip. If $Z$ contains both an entry $0$ and an entry $1$, then the zone $M[[a,b],[d,d+1]]$ must contain a~corner that intersects $(a,b,c,d)$. If all the entries of $Z$ are $1$ (i.e., $Z$ is a~slab), then by $(a,b,d,d)\in \strips_d$ and $(a,b,d+1,d+1)\notin \strips_{d+1}$ we have that one of the zones $M[[a-1,a],[d,d+1]]$ or $M[[b,b+1],[d,d+1]]$ must be a~corner that intersects $(a,b,c,d)$. Finally, if all the entries of $Z$ are $0$, then both  $M[[a-1,a],[d,d+1]]$ and $M[[b,b+1],[d,d+1]]$ are corners that intersect $(a,b,c,d)$.
\qed
\end{proof}

We now provide an algorithm that computes the canonical slab decomposition of a~given binary matrix $M$. The algorithm assumes that $M$ is provided on input via some slab decomposition, possibly not canonical.

\begin{theorem}\label{thm:canonical}
There is an algorithm $\mathtt{Decompose}(n,\mathcal{K})$, that
given
an $n \times n$ binary matrix $M$ provided by specifying its slab decomposition
$\mathcal{K}$, computes the canonical slab decomposition $\mathcal{R}$ of
$M$ in $\Oh(n+|\mathcal{K}|\log \log n)$ time.
\end{theorem}
\begin{proof}
	Recall that for $i\in [1,n]$, $\strips_i$ is the set of all $i$-strips of $M$. For convenience, within this proof we consider the elements of $\strips_i$ to be segments; that is, $[a,b]\in \strips_i$ if and only if $(a,b,i,i)$ is an $i$-strip.
	
	Our first goal is to compute, for each $i\in [1,n]$, the sets
	$A_i=\strips_i\setminus \strips_{i-1}$ and $B_i\coloneqq \strips_i\setminus \strips_{i+1}$,
	where $\strips_0 = \strips_{n+1}=\emptyset$ by convention. 
	Note that by the definition each canonical slab must start with some strip belonging to $A_i$ and end with a strip belonging to $B_j$ for some $i\leq j$, $i,j\in [1,n]$. 
	Thus, we first describe how to compute the sets $A_i$ and $B_i$ for $i \in [1, n]$ from the given slab decomposition, and later we explain how to obtain the canonical slab decomposition from them.
	
	Let us start with computing for each $i\in [1,n]$ the sets of {\em{opening}} and {\em{closing segments}} of $\cal K$ in column $i$. 
	\[O_i\coloneqq \{[a,b]~|~(a,b,i,d)\in {\cal K} \textrm{ for some $d\geq i$}\};\] \[C_i\coloneqq \{[a,b]~|~(a,b,c,i)\in {\cal K} \textrm{ for some $c\leq i$}\}.\]
	
	The sets $O_i$ and $C_i$ can be easily created in $\Oh(n + |\cal K|)$ by first initializing all of them as empty linked lists and then iterating through all elements $(a, b, c, d) \in \cal K$ and appending $[a, b]$ to $O_c$ and $C_d$. Note that in particular we have $\sum_{i=1}^n (|O_i|+|C_i|)\leq 2|\cal K|$.
	
	We describe now how by a {\em{sweep}} on the set of strips from left to right we can compute sets $B_i$ for $i\in [1,n]$. We illustrate the procedure in Figure~\ref{fig:computing_B_i}.
	We initialize an adhesive segment set data structure $\mathbb{S}$ with the universe $[1,n]$. Next, we iterate through all of the column indices $i\in [1,n]$ maintaining the following invariant: after the $i$-th iteration, the set of segments stored in $\mathbb{S}$ is equal to $\strips_{i+1}$. We initialize the iteration by invoking $\mathtt{Merge}(s)$ for each $s\in O_1$ (which always adds the segment to the structure, as they are pairwise disjoint, and possibly merges it with adjacent segments); so after this operation, $\mathbb{S}$ stores exactly $\strips_1$ and therefore the invariant is satisfied for $i=0$. In the $i$-th iteration, we need to apply updates to $\mathbb{S}$ so that the stored set of segments is modified from $\strips_i$ to $\strips_{i+1}$, and at the same time compute the set $B_i=\strips_i\setminus \strips_{i+1}$. We start with constructing a~list $L$ of all the segments of $\strips_i$ that may be possibly affected by the update. These are:
	\begin{itemize}
		\item all the intervals $s\in \strips_i$ containing some $s'\in C_i$, and
		\item all the intervals $s\in \strips_i$ adjacent to some $s'\in O_{i+1}$.
	\end{itemize}
	The list $L$ can be constructed in time $\Oh((|C_i|+|O_{i+1}|)\log \log n)$ by invoking in our data structure $\mathbb{S}$ procedure $\mathtt{Containing}(s')$ for each $s'\in C_i$ and $\mathtt{Adjacent}(s')$ for each $s'\in O_{i+1}$, and putting all the resulting segments in $L$. Note that we may filter out any duplicates in $L$ using an additional static boolean table with entries $[1,n]$ where we mark lower endpoints of segments already added to $L$; and the next attempt of insertion of an already inserted segment is ignored. By construction, we have $|L|\leq |C_i|+2|O_{i+1}|$. Having constructed $L$, we perform the following update from $\strips_i$ to $\strips_{i+1}$. First, we invoke $\mathtt{Split}(s)$ for each $s\in C_i$, deleting it from the structure. Next, we invoke $\mathtt{Merge}(s)$ for each $s\in O_{i+1}$, adding it to the structure.
	
	Finally, we compute $B_i=\strips_i\setminus \strips_{i+1}$ by iterating over all $s\in L$ and verifying, using the method $\mathtt{Containing}(s)$, whether $s$ still belongs to $\strips_{i+1}$. The set $B_i$ comprises of all the segments $s\in L$ that fail this test.
	
	\begin{figure}[h]
		\centering
		\newcommand{\h}{\cellcolor{red!22!white}}
		\newcommand{\z}{\cellcolor{gray!22!white}}
		\newcommand{\g}{\cellcolor{blue!42!white}}
		\newcommand{\y}{\cellcolor{yellow!62!white}}
		\newcommand{\q}{\cellcolor{green!32!white}}
		\newcommand{\e}{\cellcolor{black!32!white}}
		\renewcommand{\arraystretch}{1.2}
		\centering
		\begin{tabular}{|M{1em}|M{1em}|M{1em}|M{1em}|}
			\hline
			0 & 0 & 0 & \e 1 \\
			\hline
			\h 1 & \h 1 & \h 1 & 0 \\
			\hline
			0 & \y 1 & \y 1 & \y 1 \\
			\hline
			\g 1 & \g 1 & \g 1 & \g 1 \\
			\hline
		\end{tabular}
	
		\vspace{3mm}	
		\small
		\begin{tabular}{|c|c|c|c|c|c|}
			\hline
			  & $O_i$ & $C_i$ & $\strips$ in $\mathbb{S}_i$  & $L_i$ & $B_i$  \\
			\hline
			$i=1$ & $\{[2,2], [4,4]\}$ & $\varnothing$ & $\{[2,2], [4,4]\}$ &$\{[2,2], [4,4]\}$ & $\{[2,2], [4,4]\}$  \\
			\hline
			$i=2$ & $\{[3,3]\}$ & $\varnothing$ & $\{[2,4]\}$ & $\varnothing$ & $\varnothing$ \\
			\hline
			$i=3$ & $\varnothing$ & $\{[2,2]\}$ & $\{[2,4]\}$ & $\{[1,1], [2,2]\}$ & $\{[2,4]\}$ \\
			\hline
			$i=4$ &  $\{[1,1]\}$ &  $\{[1,1], [3,3], [4,4]\}$ &  $\{[1,1], [3,4]\}$ & $\{[1,1], [3,3], [4,4]\}$ & $\{[1,1], [3,4]\}$ \\	
			\hline	
		\end{tabular}
		
		\let\h\undefined
		\let\g\undefined
		\let\z\undefined
		\let\y\undefined
		\let\e\undefined
		\let\q\undefined
		\caption{Illustration for the first part of the algorithm for computing the canonical slab decomposition. The matrix is given to us by a collection of slabs (marked with distinct colors): $\mathcal{K}=\{(2,2,1,3),(3,3, 2,4), (4,4, 1,4), (1,1, 4,4)\}$. In the table we use the segment notation, proposed at the beginning of the proof. The set $L$ computed in the $i$-th round is denoted here as $L_i$ and $\mathbb{S}_i$ stands for the set of segments belonging to $\mathbb{S}$ at the beginning of the $i$-th round -- by our invariant this set is equal to $\strips_i$.}
		\label{fig:computing_B_i}
	\end{figure}	
	
	Observe that the $i$-th iteration takes time $\Oh((|C_i|+|O_{i+1}|)\log \log n)$, hence the whole procedure described above runs in total time $\Oh(n+|{\cal K}|\log \log n)$.
	
	By performing a~symmetric sweep from right to left, we may compute the sets $A_i=\strips_i\setminus \strips_{i-1}$ for all $i\in [1,n]$ within the same time complexity. Note that we have $\sum_{i=1}^n (|A_i|+|B_i|)\leq 3\sum_{i=1}^n (|C_i|+|O_i|)\leq \Oh(|\cal K|)$.
	
	Having computed the sets $A_i$ and $B_i$ for each $i\in [1,n]$, we perform the final sweep that computes the canonical slab decomposition $\cal R$. To this end, we create an auxiliary array $\mathsf{slab\_start}[1..n]$ and initialize it with $\perp$ values. For each canonical slab $(a, b, c, d)$ we have that $[a, b] \in A_c, [a, b] \in B_d$ and there does not exist a pair $(i, j)$ such that $i \in [c+1, d]$ and $[a, j] \in A_i$ or a pair $(i, j)$ such that $i \in [c, d-1]$ and $[a, j] \in B_i$.
	
	We sweep our matrix from left to right, that is, we iterate with $i$ from $1$ to $n$ and for each value of $i$ we do the following.
	First, for each $(a, b) \in A_i$ we record $\mathsf{slab\_start}[a] = i$. Then, for each $(a, b) \in B_i$ we create a canonical slab $(a, b, \mathsf{slab\_start}[a], i)$. While technically not required, for consistency we may also set $\mathsf{slab\_start}[a] = \perp$ to denote that the canonical slab with one of its corners at $(a, \mathsf{slab\_start}[a])$ was already completed.

	
	\begin{figure}[h]
		\centering
			\begin{tabular}{|c|c|c|c|c|}
			\hline
			& $A_i$ & $B_i$ & $\mathsf{slab\_start}_i$ & slabs added in $i$-th round \\
			\hline
			$i=1$ & $\{[2,2], [4,4]\}$  & $\{[2,2], [4,4]\}$ & $[\perp, 1, \perp, 1]$ & $(2,2,1,1),\ (4,4,1,1)$\\
			\hline
			$i=2$ & $\{[2,4]\}$ &  $\varnothing$ & $[\perp, 2, \perp, \perp]$ & \\
			\hline
			$i=3$ & $\varnothing$ & $\{[2,4]\}$ & $[\perp, 2, \perp, \perp]$ & $(2,4, 2,3)$\\
			\hline
			$i=4$ & $\{[1,1], [3,4]\}$ & $\{[1,1], [3,4]\}$ & $[4, \perp, 4, \perp]$ & $(1,1,4,4),\ (3,4, 4,4)$\\
			\hline
		\end{tabular}
	
		\vspace{0.3cm}
		\centering
		\newcommand{\h}{\cellcolor{red!22!white}}
		\newcommand{\z}{\cellcolor{gray!22!white}}
		\newcommand{\g}{\cellcolor{blue!42!white}}
		\newcommand{\y}{\cellcolor{yellow!62!white}}
		\newcommand{\q}{\cellcolor{green!32!white}}
		\newcommand{\e}{\cellcolor{black!32!white}}
		\renewcommand{\arraystretch}{1.2}
		\centering
		\begin{tabular}{|M{1em}|M{1em}|M{1em}|M{1em}|}
			\hline
			0 & 0 & 0 & \e 1 \\
			\hline
			\h 1 & \y 1 & \y 1 & 0 \\
			\hline
			0 & \y 1 & \y 1 & \q 1 \\
			\hline
			\g 1 & \y 1 & \y 1 & \q 1 \\
			\hline
		\end{tabular}

		\let\h\undefined
		\let\g\undefined
		\let\z\undefined
		\let\y\undefined
		\let\e\undefined
		\let\q\undefined
		\caption{The second part of the algorithm for computing the canonical slab decomposition. The sets are computed for the matrix provided in Figure~\ref{fig:computing_B_i}. $\mathsf{slab\_start}_i$ denotes the state of the array $\mathsf{slab\_start}$ after processing $A_i$, but before processing $B_i$. The matrix presents the canonical slab decomposition -- each slab is marked with a distinct color.}
		\label{fig:canonical_slab_decomposition_procedure}
	\end{figure}
	
	It is straightforward to see that the set $\cal R$ computed at the end of the iteration is exactly the canonical slab decomposition of $M$. Also, the final sweep constructing $\cal R$ runs in time $\Oh(n+|{\cal K}|)$.
	\qed
\end{proof}

By combining Lemma~\ref{lem:canonical-small} and Theorem~\ref{thm:canonical}, we arrive at the following corollary.

\begin{corollary}\label{cor:canonical}
	If $M$ is a~$d$-twin-ordered $n\times n$ binary matrix and $\cal K$ is a~slab decomposition of $M$, then the algorithm $\mathtt{Decompose}(n,\mathcal{K})$ computes in time  $\Oh(n+|\mathcal{K}|\log \log n)$ the~canonical slab decomposition of $M$ with
	$4 f_d (n+2)+4n \in \Oh_d(n)$ slabs.
\end{corollary}

\section{An amortized data structure}
\label{sec:amortized}

In this section we present a~weaker variant of the data structure promised in Theorem~\ref{thm:main}, where we allow amortized running time of the updates. We will de-amortize this data structure using standard methods in Section~\ref{sec:wcase}.

\begin{theorem}\label{thm:amortized}
There is a~data structure $\mathbb{A}$ that maintains a
$d$-twin-ordered $n\times n$ binary matrix $M$
by supporting the following operations:
\begin{enumerate}
 \item $\mathtt{Init}(n, \cal K)$: Initialize the data structure
with a~slab decomposition $\cal K$ of $M$ in
$\Oh((n + |{\cal K}|) \log \log n)$ expected time.
 \item $\mathtt{QueryBit}(i,j)$: Return the entry $M[i, j]$ in
 $\Oh(\log \log n)$ expected worst case time.
 \item $\mathtt{Update}(i, j)$: Flip the entry $M[i, j]$ in
 $\Oh(\log \log n)$ expected amortized time.
\end{enumerate}
The data structure $\mathbb{A}$ uses $\Oh_d(n)$ space.
\end{theorem}

\begin{proof}
The data structure $\mathbb{A}$ internally maintains:
\begin{enumerate}
 \item a~static 2D orthogonal point location data structure provided by Theorem~\ref{thm:Chan}, referred to as $\mathbb{P}$, storing at most
 $4 f_d(n+2)+4n \in \Oh_d(n)$ slabs with all coordinates in the range
 $\{ 1, \ldots, n \}$; and
 \item a~dictionary (implemented as a~hash table) $\mathbb{Q}$ where
 the keys are
 from universe
 $\{1, \ldots, n \}\times [1,n]$, the values are
 boolean,
 and the number of elements that will be added is bounded by
 $8 f_d(n+2) \in \Oh_d(n)$. We assume a~standard implementation
 of a~hash
 table with space linear in the number of added elements
 and constant expected query time.
\end{enumerate}

The intuition behind maintaining $\mathbb{P}$ and $\mathbb{Q}$
is the following: $\mathbb{P}$ stores a~slab decomposition of $M$
computed at some point in the past. When the next updates arrive,
instead of immediately recomputing the slab decomposition, we
store the updates in $\mathbb{Q}$ as long as the size of
$\mathbb{Q}$ is bounded by $8 f_d(n+2)$. Only when the size
of $\mathbb{Q}$ reaches this value, we recompute the up to date
slab
decomposition of $M$ and store it in $\mathbb{P}$. This allows
us to amortize for the recomputation cost.
Thus, we maintain the following invariant after each update:
If $(i, j) \in \mathbb{Q}$, then $M[i, j] = \mathbb{Q}[i, j]$.
Otherwise $M[i, j]=1$ if and only if $(i, j)$ belongs to one of the
slabs stored in $\mathbb{P}$.

\paragraph*{Initialization:} The method $\mathtt{Init}(n, \cal K)$ applies the following steps.
First, it initializes an empty dictionary $\mathbb{Q}$. Then, it
applies the algorithm $\mathtt{Decompose}(n,\mathcal{K})$ of Theorem~\ref{thm:canonical} to obtain the canonical
slab decomposition $\cal R$ of $M$. This runs in time
$\Oh(n+|\mathcal{K}| \log \log n)$, and by
Lemma~\ref{lem:canonical-small} we have $|\mathcal{R}| \leq 4 f_d(n+2)+4n \in \Oh_d(n)$.
We conclude by initializing
$\mathbb{P}$ with $\cal R$ in $\Oh_d(n \log \log n)$ time.
Thus, the total initialization time is
$\Oh((n + |{\cal K}|) \log \log n)$.

\paragraph*{Query:}
The method $\mathtt{QueryBit}(i,j)$ looks up $(i,j)$ in $\mathbb{Q}$
and upon a~successful search returns the value of $(i,j)$.
This takes $\Oh(1)$ worst case expected time.
If the search fails, a~point location query $(i,j)$ is issued to
the point location data structure $\mathbb{P}$. This takes
$\Oh(\log \log n)$ worst case time. If the point location
query is successful, we return $1$, otherwise we return $0$.

\paragraph*{Update:}
The method $\mathtt{Update}(i, j)$ looks up $(i,j)$ in
$\mathbb{Q}$. If it finds it, then it negates its value.
This takes constant time in expectation.
Otherwise it issues a~point location query $(i,j)$ to $\mathbb{P}$
($\Oh(\log \log n)$ worst case time)
and adds $(i,j)$ to $\mathbb{Q}$ with the negated result of the query
(constant time in expectation).
Finally,
if the size of $\mathbb{Q}$ reaches the bound of $8 f_d(n+2)$,
the method
recomputes $\mathbb{P}$ and sets $\mathbb{Q}$ to be empty for
accommodating the next $8 f_d(n+2)$ updates.
We now describe the process of recomputing $\mathbb{P}$ (see Figure~\ref{fig:amortized_procedure} for an illustration).

In essence, recomputing $\mathbb{P}$ boils down to finding
a slab decomposition $\cal K$ that describes the
matrix $M$ after all updates stored in $\mathbb{Q}$ have been
introduced. Once we have such a~decomposition $\cal K$ at hand,
we use the method $\mathtt{Decompose}(n,\cal K)$ to compute
the canonical decomposition $\mathcal{R}$ in time
$\Oh(n+|\mathcal{K}| \log \log n)$ by
Theorem~\ref{thm:canonical}. As $|{\cal R}|\leq 4 f_d(n+2)+4n \in \Oh_d(n)$ by Lemma~\ref{lem:canonical-small}, we may then re-initialize $\mathbb{P}$ with ${\cal R}$
in time $\Oh_d(n \log \log n)$, by Theorem~\ref{thm:Chan}.
So let us proceed to describing how to find the decomposition
$\cal K$ of small enough size, given the list of slabs $\mathcal{P}$
from $\mathbb{P}$ (that can be computed upon the last recomputation
of
$\mathbb{P}$ and stored since then until now) and the list of
updates $\cal Q$ from $\mathbb{Q}$, that can be extracted from
$\mathbb{Q}$ in time proportional to the size of $\mathbb{Q}$,
bounded by $8 f_d(n+2)$.

\begin{figure}
	\newcommand{\h}{\cellcolor{red!22!white}}
	\newcommand{\z}{\cellcolor{gray!22!white}}
	\newcommand{\g}{\cellcolor{blue!42!white}}
	\newcommand{\y}{\cellcolor{yellow!62!white}}
	\newcommand{\q}{\cellcolor{green!32!white}}
	\newcommand{\e}{\cellcolor{black!32!white}}
	\renewcommand{\arraystretch}{1.2}
	\centering
	\begin{tabular}{|M{1em}|M{1em}|M{1em}|M{1em}|M{1em}|}
		\hline
		0 & {\bf 0} & 0 & \e 1 & 0\\
		\hline
		0 & 0 & 0 & \e 1 & 0 \\
		\hline
		\g 1 & \g 1 & \y 1 & 0 & 0 \\
		\hline
		\g 1 & \g 1 & \y 1 & {\bf 0} & \q 1 \\
		\hline
		\g 1 & \g {\bf 1} & 0 & \h 1 & \q 1 \\
		\hline
	\end{tabular}
	\hspace{0.5cm}
	\begin{tabular}{|M{1em}|M{1em}|M{1em}|M{1em}|M{1em}|}
		\hline
		0 & \z 1 & 0 & \e 1 & 0\\
		\hline
		0 & 0 & 0 & \e 1 & 0 \\
		\hline
		\g 1 & \y 1 & \y 1 & 0 & 0 \\
		\hline
		\g 1 & \y 1 & \y 1 & \h 1 & \h 1 \\
		\hline
		\g 1 &  0 & 0 & \h 1 & \h 1 \\
		\hline
	\end{tabular}
	\let\h\undefined
	\let\g\undefined
	\let\z\undefined
	\let\y\undefined
	\let\e\undefined
	\let\q\undefined
	\caption{A collection of slabs in $\cal P$ and the updates in $\cal Q$ (marked with bold values). In the second figure, there is a canonical slab decomposition after the recomputation.}
	\label{fig:amortized_procedure}
\end{figure}

The updates in $\cal Q$ can be partitioned into $0$-updates, i.e.,
the entries $(i,j)$ in ${\cal Q}$ whose value is $0$, and
$1$-updates, i.e., the entries $(i,j)$ in $\cal Q$ whose value is $1$.
We first sort $\cal Q$ lexicographically in $\Oh(|{\cal Q}|)$ time.
Using $\mathbb{P}$, we assign the $0$-updates
to slabs of $\cal P$ that contain them.
As a~result, each slab $P \in \mathcal{P}$ is assigned a~list
$L_P$ of $0$-updates from $\cal Q$, each of which is contained in $P$,
sorted lexicographically.
Next, for each
slab $P \in \mathcal{P}$, we iterate
through $L_P$ successively splitting $P$ into smaller
slabs.
We will process $L_P$ in batches, each batch corresponding
to updates of $L_P$ with equal first coordinate.
Thus, let $L_i$ be the sorted list of
second coordinates of updates from $L_P$ with the first
coordinate $i$.
Note that $(i,k) \in P$ for all $k \in L_i$.
Let $P=(a,b,c,d)$.
For each $i \in \{ 1, \ldots, n \}$ with $L_i \neq \emptyset$,
we proceed as follows.

We first partition the slab $P$ into three slabs:
$P[<i] = (a,i-1,c,d)$,
$P[i]=(i,i,c,d)$, and
$P[>i]= (i+1,b,c,d)$. The slab $P[<i]$, if it is
non-empty, is directly added to $\mathcal{K}$ as there are
no updates within it that we need to process. We now need to
split $P[i]$ into smaller slabs according to $L_i$, and then we proceed with
processing
$P\coloneqq P[>i]$ in the next step of the iteration. What remains
is to explain how $P[i]$ is split with regard to $L_i$. This is
done in a~very simple fashion: if $L_i=\{ k_1,\ldots,k_\ell \}$ with $k_1\leq \ldots\leq k_\ell$, then
we split $P[i]$ into
$P_0[i]=(i,i,c,k_1-1), P_2[i]=(i,i,k_1+1,k_2-1),
\ldots P_{\ell}[i]=(i,i,k_\ell+1,d)$, and we add each non-empty
slab
among $P_0[i],P_1[i], \ldots, P_\ell[i]$ to $\mathcal{K}$.
This entire iteration step takes $\Oh(|L_i|)$ time and
creates $\Oh(|L_i|)$ new slabs that are added to $\cal K$.

After we are done with $0$-updates from $\mathcal{Q}$,
we still have to process the $1$-updates. For
every $1$-update in $\mathcal{Q}$, we first use
structure $\mathbb{P}$ to see whether this update is contained
in some slab stored in $\mathbb{P}$. If it is, we ignore it,
otherwise we add to $\cal K$ a~$1 \times 1$ slab with
coordinates of the update that is being processed.

The entire procedure adds to $\cal K$ at most
$|\mathcal{P}|+2|\mathcal{Q}| \leq 4 f_d(n+2)+4n + 16 f_d(n+2)$ slabs.
Thus, the procedure $\mathtt{Decompose}(n,\cal K)$
runs in $\Oh(( f_d (n+2))\log \log n)$ time, and $\mathbb{P}$
can be re-initialized in $\Oh(( f_d (n+2))\log \log n)$ time as well.
Amortized over $8 f_d(n+2)$ updates, this gives
an amortized expected time of $\Oh(\log \log n)$ per update operation.
\qed\end{proof}

\section{De-amortization}
\label{sec:wcase}

In this section we show how to de-amortize the data structure
of
Theorem~\ref{thm:amortized}. Precisely, we prove the following statement, which immediately implies Theorem~\ref{thm:main}.

\begin{theorem}\label{thm:wc}
There is a data structure $\mathbb{W}$, which maintains a
$d$-twin-ordered $n\times n$ binary matrix $M$
by supporting the following operations:
\begin{enumerate}
 \item $\mathtt{Init}(n, \cal K)$: Initialize the data structure
with a slab decomposition $\cal K$ of $M$ in
$\Oh((n + |{\cal K}|) \log \log n)$ time.
 \item $\mathtt{QueryBit}(i,j)$: Return the entry $M[i, j]$ in
 $\Oh(\log \log n)$ expected worst case time.
 \item $\mathtt{Update}(i, j)$: Flip the entry $M[i, j]$ in
 $\Oh(\log \log n)$ expected worst case time.
\end{enumerate}
The data structure $\mathbb{W}$ uses $\Oh_d(n)$ space.
\end{theorem}
\begin{proof}
We apply a standard de-amortization technique to the data structure
$\mathbb{A}$ of Theorem~\ref{thm:amortized}. For the simplicity of exposition, we firstly present a slightly flawed simpler version of the argument, and then we explain what was the issue and how to fix it.

Let $L=8 f_d(n+2)$ be the limit on the size of $\mathbb{Q}$ from Theorem~\ref{thm:amortized} and let us partition all updates into epochs of $L$ updates. Let us denote the $i$th epoch as $E_i$, with the first epoch being~$E_0$. For each valid $i \ge 1$, let us define the data structure $\mathbb{A}_i$ that will exist throughout epochs $E_{2i-2}, E_{2i-1}, E_{2i}, E_{2i+1}$ and the data structure $\mathbb{A}_0$ that will exist throughout epochs $E_0$ and $E_1$ (in particular, there will be only two different data structures existing at any point of time). The data structure $\mathbb{A}_0$ will be initialized with the $\mathtt{Init}(n, \cal K)$ call and responsible for answering queries of epochs $E_0$ and $E_1$. The data structure $\mathbb{A}_i$ for $i \ge 1$ will be initialized throughout epochs $E_{2i-2}$ and $E_{2i-1}$ and responsible for answering queries during epochs $E_{2i}$ and $E_{2i+1}$. At the start of epoch $E_{2i-2}$, we run $\mathtt{Init}(n, \cal K)$, where $\cal K$ is the slab decomposition obtained from the current state of $\mathbb{A}_{i-1}$, similarly as in the case of recomputation of the data structure in the proof of Theorem~\ref{thm:amortized}. We have that $|\mathcal{K}| = \Oh(L)$, so that call takes $\Oh(L \log \log n)$ total time. Hence, we can distribute this work over $L$ updates of epoch $E_{2i-2}$ adding $\Oh(\log \log n)$ worst case time overhead to each update. At the end of $E_{2i-2}$ the $\mathtt{Init}(n, \cal K)$ call will be finished, but $\mathbb{A}_i$ will consequently be $L$ updates behind the current state, as its state reflects $M$ before $E_{2i-2}$. During the epoch $E_{2i-1}$ the data structure $\mathbb{A}_i$ catches up by processing two updates when $\mathbb{A}_{i-1}$ processes one. That multiplies the work on each update only by an absolute constant. At the end of $E_{2i-1}$, $\mathbb{A}_i$ caught up to the current state and is ready to take over the role of being the data structure that actually answers incoming queries.

The above description was slightly flawed, as the $\mathtt{Init}(n, \cal K)$ call initializing $\mathbb{A}_{i}$ needs to know $\cal K$. However, it takes $\Oh(|\mathcal{K}|) = \Oh(L)$ time to create $\mathbb{A}_i$ and it requires uninterrupted access to the internal data of $\mathbb{A}_{i-1}$. However $\mathbb{A}_{i-1}$ needs to be answering new queries, from the epoch $E_{2i-2}$, so it cannot grant uninterrupted access to its internal $\mathbb{P}$ and $\mathbb{Q}$ for the sake of creating $\cal K$. The trick to mitigate this issue is to create two copies of each data structure --- $\mathbb{A}_{i, 0}$ and $\mathbb{A}_{i, 1}$. Up to the moment of starting $E_{2i}$ epoch, they are initialized in the same way as $\mathbb{A}_i$ was in the previous description. However, at this moment $\mathbb{A}_{i, 0}$ takes the responsibility of answering incoming updates (and consequently, requires changes to its internal data), while $\mathbb{A}_{i, 1}$ does not process any updates and serves the role of the snapshot of $\mathbb{A}_i$ after its full initialization, lending itself to $\mathbb{A}_{i+1}$ that required uninterrupted access to $\mathbb{A}_i$ for its initialization.
\qed
\end{proof}

\bibliographystyle{splncs04}
\bibliography{references}

@article{VANEMDEBOAS197780,
title = {Preserving order in a forest in less than logarithmic time and linear space},
journal = {Information Processing Letters},
volume = {6},
number = {3},
pages = {80-82},
year = {1977},
issn = {0020-0190},
doi = {10.1016/0020-0190(77)90031-X},
author = {P. {van Emde Boas}}
}

@article{Chan2DPL,
author = {Chan, Timothy M.},
title = {Persistent Predecessor Search and Orthogonal Point Location on the Word RAM},
year = {2013},
issue_date = {June 2013},
publisher = {Association for Computing Machinery},
address = {New York, NY, USA},
volume = {9},
number = {3},
issn = {1549-6325},
doi = {10.1145/2483699.2483702},
journal = {ACM Trans. Algorithms},
month = jun,
articleno = {22},
numpages = {22}
}

@inproceedings{tww5,
	author       = {{\'{E}}douard Bonnet and
	Ugo Giocanti and
	Patrice Ossona de Mendez and
	St{\'{e}}phan Thomass{\'{e}}},
	title        = {Twin-Width {V:} {L}inear Minors, Modular Counting, and Matrix Multiplication},
	booktitle    = {40th International Symposium on Theoretical Aspects of Computer Science,
	{STACS} 2023},
	series       = {LIPIcs},
	volume       = {254},
	pages        = {15:1--15:16},
	publisher    = {Schloss Dagstuhl --- Leibniz-Zentrum f{\"{u}}r Informatik},
	year         = {2023},
	doi          = {10.4230/LIPICS.STACS.2023.15},
	bibsource    = {dblp computer science bibliography, https://dblp.org}
}

@inproceedings{tww-distal,
	author       = {Wojciech Przybyszewski},
	title        = {Distal Combinatorial Tools for Graphs of Bounded Twin-Width},
	booktitle    = {38th Annual {ACM/IEEE} Symposium on Logic in Computer Science, {LICS}
	2023},
	pages        = {1--13},
	publisher    = {{IEEE}},
	year         = {2023},
	doi          = {10.1109/LICS56636.2023.10175719},
	bibsource    = {dblp computer science bibliography, https://dblp.org}
}

@article{tww-pchi,
	author = {Bourneuf, Romain and Thomass{\' e}, St{\' e}phan},
	journal = {Advances in Combinatorics},
	doi = {10.19086/aic.2025.2},
	year = {2025},
	number = {2},
	title = {Bounded twin-width graphs are polynomially $\chi$-bounded},
}

@article{tww-qchi,
	author       = {Michał Pilipczuk and
	Marek Sokołowski},
	title        = {Graphs of bounded twin-width are quasi-polynomially $\chi$-bounded},
	journal      = {J. Comb. Theory {B}},
	volume       = {161},
	pages        = {382--406},
	year         = {2023},
	doi          = {10.1016/J.JCTB.2023.02.006},
	bibsource    = {dblp computer science bibliography, https://dblp.org}
}

@inproceedings{tww8,
	author       = {{\'{E}}douard Bonnet and
	Dibyayan Chakraborty and
	Eun Jung Kim and
	Noleen K{\"{o}}hler and
	Raul Lopes and
	St{\'{e}}phan Thomass{\'{e}}},
	title        = {Twin-Width {VIII:} {D}elineation and Win-Wins},
	booktitle    = {17th International Symposium on Parameterized and Exact Computation,
	{IPEC} 2022},
	series       = {LIPIcs},
	volume       = {249},
	pages        = {9:1--9:18},
	publisher    = {Schloss Dagstuhl --- Leibniz-Zentrum f{\"{u}}r Informatik},
	year         = {2022},
	doi          = {10.4230/LIPICS.IPEC.2022.9},
	bibsource    = {dblp computer science bibliography, https://dblp.org}
}

@inproceedings{tww6,
	author       = {{\'{E}}douard Bonnet and
	Eun Jung Kim and
	Amadeus Reinald and
	St{\'{e}}phan Thomass{\'{e}}},
	title        = {Twin-width {VI:} the lens of contraction sequences},
	booktitle    = {2022 {ACM-SIAM} Symposium on Discrete Algorithms,
	{SODA} 2022},
	pages        = {1036--1056},
	publisher    = {{SIAM}},
	year         = {2022},
	doi          = {10.1137/1.9781611977073.45},
	bibsource    = {dblp computer science bibliography, https://dblp.org}
}

@article{tww7,
	author       = {{\'{E}}douard Bonnet and
	Colin Geniet and
	Romain Tessera and
	St{\'{e}}phan Thomass{\'{e}}},
	title        = {Twin-width {VII:} groups},
	journal      = {CoRR},
	volume       = {abs/2204.12330},
	year         = {2022},
	doi          = {10.48550/ARXIV.2204.12330},
	eprinttype    = {arXiv},
	eprint       = {2204.12330},
	bibsource    = {dblp computer science bibliography, https://dblp.org}
}

@inproceedings{tww-tournaments,
	author       = {Colin Geniet and
	St{\'{e}}phan Thomass{\'{e}}},
	title        = {First Order Logic and Twin-Width in Tournaments},
	booktitle    = {31st Annual European Symposium on Algorithms, {ESA} 2023},
	series       = {LIPIcs},
	volume       = {274},
	pages        = {53:1--53:14},
	publisher    = {Schloss Dagstuhl --- Leibniz-Zentrum f{\"{u}}r Informatik},
	year         = {2023},
	doi          = {10.4230/LIPICS.ESA.2023.53},
	bibsource    = {dblp computer science bibliography, https://dblp.org}
}

@inproceedings{tww-separable,
	author       = {{\'{E}}douard Bonnet and
	Romain Bourneuf and
	Colin Geniet and
	St{\'{e}}phan Thomass{\'{e}}},
	title        = {Factoring Pattern-Free Permutations into Separable ones},
	booktitle    = {2024 {ACM-SIAM} Symposium on Discrete Algorithms,
	{SODA} 2024},
	pages        = {752--779},
	publisher    = {{SIAM}},
	year         = {2024},
	doi          = {10.1137/1.9781611977912.30},
	bibsource    = {dblp computer science bibliography, https://dblp.org}
}

@inproceedings{tww-repr,
	author       = {Michał Pilipczuk and
	Marek Sokołowski and
	Anna Zych{-}Pawlewicz},
	title        = {Compact Representation for Matrices of Bounded Twin-Width},
	booktitle    = {39th International Symposium on Theoretical Aspects of Computer Science,
	{STACS} 2022},
	series       = {LIPIcs},
	volume       = {219},
	pages        = {52:1--52:14},
	publisher    = {Schloss Dagstuhl --- Leibniz-Zentrum f{\"{u}}r Informatik},
	year         = {2022},
	doi          = {10.4230/LIPICS.STACS.2022.52},
	bibsource    = {dblp computer science bibliography, https://dblp.org}
}

@inproceedings{tww-stable,
	author       = {Jakub Gajarsk{\'{y}} and
	Michał Pilipczuk and
	Szymon Toruńczyk},
	title        = {Stable graphs of bounded twin-width},
	booktitle    = {37th Annual {ACM/IEEE} Symposium on Logic in Computer
	Science, LICS 2022},
	pages        = {39:1--39:12},
	publisher    = {{ACM}},
	year         = {2022},
	doi          = {10.1145/3531130.3533356},
	bibsource    = {dblp computer science bibliography, https://dblp.org}
}

@inproceedings{tww-types,
	author       = {Jakub Gajarsk{\'{y}} and
	Michał Pilipczuk and
	Wojciech Przybyszewski and
	Szymon Toruńczyk},
	title        = {Twin-Width and Types},
	booktitle    = {49th International Colloquium on Automata, Languages, and Programming,
	{ICALP} 2022},
	series       = {LIPIcs},
	volume       = {229},
	pages        = {123:1--123:21},
	publisher    = {Schloss Dagstuhl - Leibniz-Zentrum f{\"{u}}r Informatik},
	year         = {2022},
	doi          = {10.4230/LIPICS.ICALP.2022.123},
	bibsource    = {dblp computer science bibliography, https://dblp.org}
}

@article{tww4,
	author       = {{\'{E}}douard Bonnet and
	Ugo Giocanti and
	Patrice {Ossona de Mendez} and
	Pierre Simon and
	St{\'{e}}phan Thomass{\'{e}} and
	Szymon Toruńczyk},
	title        = {Twin-Width {IV:} {O}rdered Graphs and Matrices},
	journal      = {J. {ACM}},
	volume       = {71},
	number       = {3},
	pages        = {21},
	year         = {2024},
	doi          = {10.1145/3651151},
	bibsource    = {dblp computer science bibliography, https://dblp.org}
}

@article{tww-kern,
	author       = {{\'{E}}douard Bonnet and
	Eun Jung Kim and
	Amadeus Reinald and
	St{\'{e}}phan Thomass{\'{e}} and
	R{\'{e}}mi Watrigant},
	title        = {Twin-width and Polynomial Kernels},
	journal      = {Algorithmica},
	volume       = {84},
	number       = {11},
	pages        = {3300--3337},
	year         = {2022},
	doi          = {10.1007/S00453-022-00965-5},
	bibsource    = {dblp computer science bibliography, https://dblp.org}
}

@article{tww-perm,
	author       = {{\'{E}}douard Bonnet and
	Jaroslav Ne\v{s}et\v{r}il and
	Patrice {Ossona de Mendez} and
	Sebastian Siebertz and
	St{\'{e}}phan Thomass{\'{e}}},
	title        = {Twin-width and permutations},
	journal      = {Log. Methods Comput. Sci.},
	volume       = {20},
	number       = {3},
	year         = {2024},
	doi          = {10.46298/LMCS-20(3:4)2024},
	bibsource    = {dblp computer science bibliography, https://dblp.org}
}

@article{tww2,
	author       = {{\'{E}}douard Bonnet and
	Colin Geniet and
	Eun Jung Kim and
	St{\'{e}}phan Thomass{\'{e}} and
	R{\'{e}}mi Watrigant},
	title        = {Twin-width {II:} small classes},
	journal      = {Comb. Theory},
	volume       = {2},
	number       = {2},
	year         = {2022},
	doi          = {10.5070/C62257876},
	bibsource    = {dblp computer science bibliography, https://dblp.org}
}

@article{tww1,
	author       = {{\'{E}}douard Bonnet and
	Eun Jung Kim and
	St{\'{e}}phan Thomass{\'{e}} and
	R{\'{e}}mi Watrigant},
	title        = {Twin-width {I:} {T}ractable {FO} Model Checking},
	journal      = {J. {ACM}},
	volume       = {69},
	number       = {1},
	pages        = {3:1--3:46},
	year         = {2022},
	doi          = {10.1145/3486655},
	bibsource    = {dblp computer science bibliography, https://dblp.org}
}

\end{document}